\documentclass[conference]{IEEEtran}
\IEEEoverridecommandlockouts

\usepackage{amsmath}
\usepackage{amssymb}
\usepackage{amsthm}
\usepackage{mathtools}
\usepackage{color}
\usepackage{wrapfig}
\usepackage{array}
\usepackage{scrextend}
\usepackage{multirow}
\usepackage{xcolor}
\usepackage{tikz}
\usepackage[margin=0.73in,top=0.437in]{geometry} 
\usetikzlibrary{shapes,arrows,positioning}

\definecolor{green1}{RGB}{81,107,62}
\definecolor{blue1}{RGB}{0,32,96}
\tikzstyle{block} = [rectangle, draw, text centered, rounded corners, minimum height=2em]

\newtheorem{thm}{Theorem}
\newtheorem{defn}{Definition}
\newtheorem{lemma}{Lemma}

\newtheorem{construction}{Construction}
\newtheorem{remark}{Remark}

\newcommand{\ve}[1]{\ensuremath{\mathbf{#1}}}
\newcommand{\e}{\ensuremath{\mathrm{e}}}

\newcolumntype{C}[1]{>{\centering\let\newline\\\arraybackslash\hspace{0pt}}m{#1}}

\usepackage{color}

\begin{document}

\title{{\textbf{Coding over Sets for DNA Storage}}}

\author{
	\IEEEauthorblockN{
		\textbf{Andreas Lenz}\IEEEauthorrefmark{1},
		\textbf{Paul H. Siegel}\IEEEauthorrefmark{2},
		\textbf{Antonia Wachter-Zeh}\IEEEauthorrefmark{1}, and
		\textbf{Eitan Yaakobi}\IEEEauthorrefmark{3}
	}

	\IEEEauthorblockA{
		\IEEEauthorrefmark{1}Institute for Communications Engineering, Technical University of Munich, Germany
	}
	\IEEEauthorblockA{
		\IEEEauthorrefmark{2}Department of Electrical and Computer Engineering, CMRR, University of California, San Diego, California
	}
	\IEEEauthorblockA{
		\IEEEauthorrefmark{3}Computer Science Department, Technion -- Israel Institute of Technology, Haifa, Israel
	}
	\textbf{Emails}: andreas.lenz@mytum.de, psiegel@ucsd.edu, antonia.wachter-zeh@tum.de, yaakobi@cs.technion.ac.il
	
}\vspace{-1ex}

\maketitle
\thispagestyle{plain}
\pagestyle{plain}
\begin{abstract}
In this paper, we study error-correcting codes for the storage of data in synthetic deoxyribonucleic acid (DNA). We investigate a storage model where data is represented by an unordered set of $M$ sequences, each of length~$L$. Errors within that model are losses of whole sequences and point errors inside the sequences, such as substitutions, insertions and deletions. We propose code constructions which can correct these errors with efficient encoders and decoders. By deriving upper bounds on the cardinalities of these codes using sphere packing arguments, we show that many of our codes are close to optimal.
\end{abstract}\vspace{-.5ex}

\IEEEpeerreviewmaketitle

\section{Introduction}
\vspace{0ex}
DNA-based storage has attracted significant attention due to recent demonstrations of the viability of storing information in macromolecules. This recent increased interest was paved by significant progress in synthesis and sequencing technologies. The main advantages of DNA-based storages over classical storage technologies are very high data densities and long-term reliability without electrical supply. Given the trends in cost decreases of DNA synthesis and sequencing, it is now acknowledged that within the next $10$--$15$ years DNA storage may become a highly competitive archiving technology. 

A DNA storage system consists of three important entities (see Fig.~\ref{fig:DNAStorage}): (1) a DNA synthesizer that produces the strands that encode the data to be stored in DNA. In order to produce strands with acceptable error rate the length of the strands is typically limited to no more than 250 nucleotides; (2) a storage container with compartments that store the DNA strands, although in an unordered manner; (3) a DNA sequencer that reads the strands and transfers them back to digital data. The encoding and decoding stages are external processes to the storage system which convert the binary user data into strands of DNA in such a way that even in the presence of errors, it is possible to reconstruct the original data.

The first large scale experiments that demonstrated the potential of \emph{in vitro} DNA storage were reported by Church et al. who recovered 643 KB of data~\cite{CGK12} and Goldman et al. who accomplished the same task for a 739 KB message~\cite{GBCDLSB13}. Later, in~\cite{GHPPS15}, Grass et al. stored and recovered successfully an 81 KB message by using error-correcting codes. Since then, several groups have built similar systems, storing ever larger amounts of data.  Among these, Erlich and Zielinski~\cite{EZ17}  stored 2.11MB of data with high storage rate,  Blawat et al.~\cite{BGHCTIPC16} successfully stored 22MB, and more recently Organick et al.~\cite{Oetal17} stored 200MB. Yazdi et al.~\cite{YTMZM15} developed a method that offers both random access and rewritable storage.

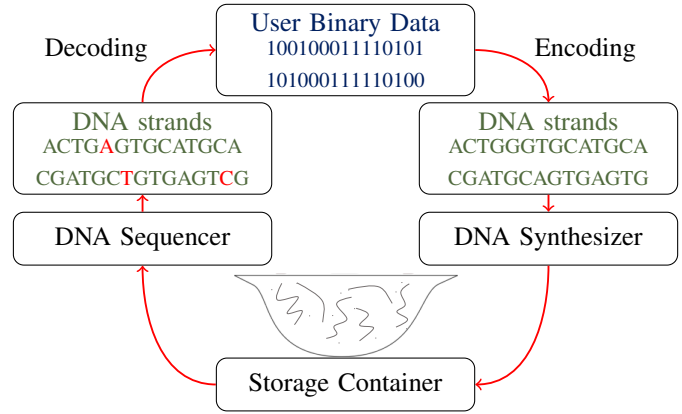
\begin{figure}%
	\centering
	\vspace{-1.5ex}

	\begin{tikzpicture}[node distance = 3cm, auto, text width = 3.2cm]
	
	\node [block, align=center]  (data) {\textcolor{blue1}{User Binary Data\\\footnotesize100100011110101\\\footnotesize101000111110100}};

	\node [block, align=center, below of= data, node distance = 4.45cm] (con) {Storage Container};
	\node [above left=-0.11cm and -3.55cm of con] (con2) {\includegraphics[width=3.0cm]{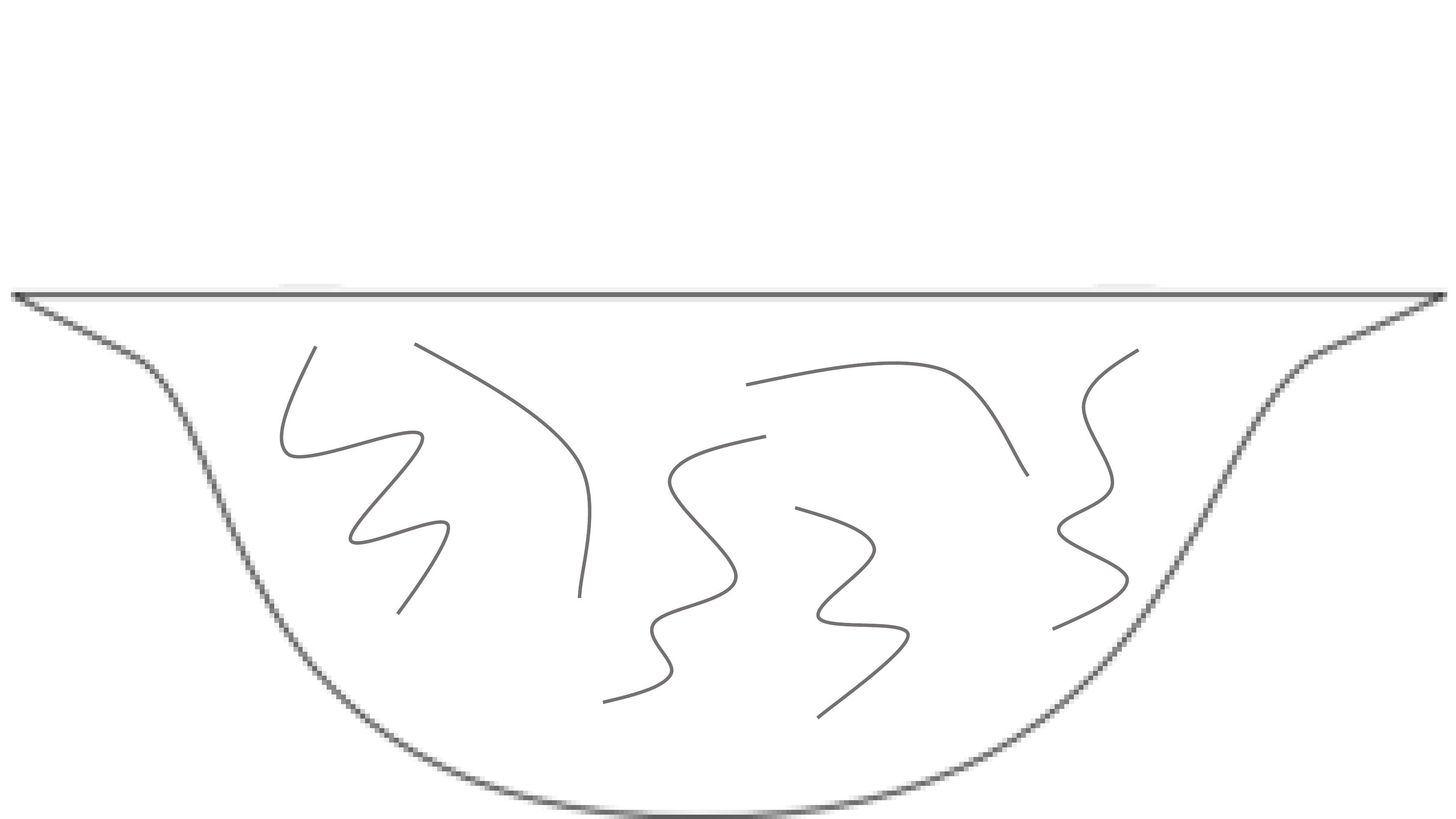}};

	\node [block, align=center, below right=0.08cm and -0.75cm of data, node distance = 2.5cm] (strands) {\textcolor{green1}{DNA strands\\\footnotesize ACTGGGTGCATGCA \\ \footnotesize CGATGCAGTGAGTG} };
	
	\node [block, align=center, below of= strands, node distance = 1.2cm] (synth) {DNA Synthesizer};
	
	\node [block, align=center, below left=0.08cm and -0.75cm of data, node distance = 2.5cm] (strands2) {\textcolor{green1}{DNA strands\\\footnotesize ACTG\textcolor{red}{A}GTGCATGCA \\ \footnotesize CGATGC\textcolor{red}{T}GTGAGT\textcolor{red}{C}G} };
	
	\node [block, align=center, below of= strands2, node distance = 1.2cm] (seq) {DNA Sequencer};
	
	\draw [->, color=red,line width=0.25mm] (data.east) to [out=0,in=90] (strands.north) ;
	\draw [->, color=red,line width=0.25mm] (strands.south) to (synth.north);
	\draw [->, color=red,line width=0.25mm] (synth.south) to [out=270,in=0] (con.east);
	\draw [->, color=red,line width=0.25mm] (con.west) to [out=180,in=270] (seq.south);
	\draw [->, color=red,line width=0.25mm] (seq.north) to (strands2.south);	
	\draw [->, color=red,line width=0.25mm] (strands2.north) to [out=90,in=180] (data.west);
	
	\node[] at (strands |- data) {\hspace{1.3cm} Encoding};
	\node[] at (strands2 |- data) {\hspace{0.3cm}Decoding};
	
	\end{tikzpicture}
	
	\vspace{-1.0ex}\caption[DNA Storage system]
	{Illustration of a DNA-based storage system.}\vspace{-3.5ex}
	\label{fig:DNAStorage}%
\end{figure}

DNA as a storage system has several attributes which distinguish it from any other storage system. The most prominent one is that the strands are not ordered in the memory and thus it is not possible to know the order in which they were stored. One way to address this problem is using block addresses, also called indices, that are stored as part of the strand. Errors in DNA are typically substitutions, insertions, and deletions, where most published studies report that either substitutions or deletions are the most common ones, depending upon the specific technology for synthesis and sequencing~\cite{BGHCTIPC16,EZ17,KC14,Oetal17,RRCHLHNJ13,YYLWLCKC14}. While codes correcting substitution errors were widely studied, much less is known for codes correcting insertions and deletions. The task of error correction becomes even more challenging taking into account the lack of ordering of the strands. The goal of this work is to study and to design error-correcting codes which are specifically targeted towards the special structure of DNA storage systems.

\section{DNA Channel Model} \label{sec:channel:model}
We consider the storage of data in synthetic DNA and build upon the analysis of \cite{HRT17} and \cite{KPM16}. In such a system, data is stored as an unordered \emph{set}
\begin{equation*}
\mathcal{S} = \{ \mathbf{x}_1, \mathbf{x}_2, \dots, \mathbf{x}_M \} \subseteq \mathbb{Z}_2^L,
\end{equation*}
with distinct \textit{sequences} $\mathbf{x}_i$. The parameter $M$ denotes the number of stored sequences and $L$ is the length of each sequence $\mathbf{x}_i$. The set of all possible data sets is therefore
\begin{equation*}
	\mathcal{X}_M^L = \{\mathcal{S} \subseteq \mathbb{Z}_2^L : |\mathcal{S}| = M\},
\end{equation*}
and note that $|\mathcal{X}_M^L| = \binom{2^L}{M}$.
Representing data words as unordered sets is inherently natural, as any information about ordering of the data sequences is lost during the storage.

When a data set $\mathcal{S}\in \mathcal{X}_M^L$ is read from the storage medium, a subset of $M-s$ sequences is obtained, of which additionally~$t$ are erroneous.   This received data set $\mathcal{S}'$ is considered to be the output of a pre-processing algorithm, which produces estimates of the stored sequences with a reconstruction algorithm.

Denote by $(\mathcal{U}, \mathcal{L}, \mathcal{F})$ a partition of $\mathcal{S}$ such that: 
\begin{itemize}
	\item the set $\mathcal{U}$ corresponds to the $M-t-s$ sequences that have been received without errors,
	\item $\mathcal{L}$ is the set of $s$ sequences that have not been read at all,
	\item $\mathcal{F}$ is the set of $t$ sequences that are read with errors.
\end{itemize}
Hence, the channel output is $\mathcal{S}' = \mathcal{U} \cup \mathcal{F}'$, where $\mathcal{F}'=\{\ve{x}'_{f_1}, \dots, \ve{x}'_{f_t} \}$ is the set of received sequences that are in error. In each erroneous sequence $\ve{x}'_{f_i}$, there are at most $\epsilon$ substitution or insertion and deletion errors. Note that in contrast to \cite{KT18}, where the storage of \textit{multisets} with full sequence errors ($\epsilon=L$) are discussed, we consider unordered \textit{sets} with point errors ($\epsilon \leq L$). Since the erroneous sequences $\ve{x}_{f_i}'$ are not necessarily distinct from each other or from the sequences in $\mathcal{U}$, the size of the received set satisfies $M-t-s \leq |\mathcal{S}'| \leq M-s$. In view of our channel model, we will refer to the following definition of an error-correcting code for a DNA storage system.
\begin{defn}
	A code $\mathcal{C} \subseteq \mathcal{X}_M^L$ is called an \textbf{$(s,t,\epsilon)_\mathcal{H}$ error-correcting code}, if it can correct $s$ (or fewer) losses of sequences and $\epsilon$ (or fewer) substitutions in each of $t$ (or fewer) sequences. Similarly, an \textbf{$(s,t,\epsilon)_\mathcal{L}$ error-correcting code} is defined for insertion/deletion errors, where the erroneous sequences have at most $\epsilon$ insertions and deletions.
\end{defn}
Here, the subscripts $\mathcal{H}, \mathcal{L}$ abbreviate the underlying Hamming, respectively Levenshtein metric. With this definition a code $\mathcal{C}$ is a set of codewords, where each codeword is itself a set of $M$ sequences of length $L$. One of the main challenges associated with errors in such codewords is the loss of ordering of the code sequences. Throughout the paper we will use the following definition for the redundancy of a code.
\begin{defn}
	The \emph{redundancy} of a code $\mathcal{C} \subseteq \mathcal{X}_M^L$ is \vspace{-0.5ex}
	$$ r(\mathcal{C}) = \log|\mathcal{X}_M^L| - \log |\mathcal{C}| =  \log\binom{2^L}{M} - \log |\mathcal{C}|.\vspace{-0.5ex}$$
\end{defn}
Here and in the rest of the paper, we take the logarithm to the base $2$. We summarize several comments on the DNA storage model in the following remark.
\begin{remark}
\begin{addmargin}[-1.2em]{0em}
\begin{enumerate}
\item We choose to work here with sets and not multisets of sequences because sequences are assumed to be replicated prior to reading, and the reading process does not necessarily recover all of the copies. Thus it is not possible to distinguish how many times each sequence was stored. For more details, see~\cite{HRT17}.
\item Even though there is no order of the sequences in the set $\mathcal{S} = \{ \mathbf{x}_1, \mathbf{x}_2, \dots, \mathbf{x}_M \}$, for notational purposes we assume they are listed in lexicographic order in the set representation of $\mathcal{S}$. However, this ordering information is not available when reading the sequences. A common and efficient solution to combat the lack of ordering of the sequences is to add an index for each sequence~\cite{YTMZM15,HRT17}. This requires an index of $\lceil \log M\rceil$ bits in each sequence, which limits the maximum number of information bits to be $M(L-\lceil \log M\rceil)$. While this solution is attractive for its simplicity, it introduces already a redundancy of \vspace{-0.5ex}
$$   \log\binom{2^L}{M} - M(L-\lceil \log M\rceil) = c_MM,\vspace{-0.5ex}$$
where $c_M = (\lceil \log M \rceil - \log M) + \log \e - \nu$ with $\nu = o(1)$ and $\nu \leq 1+\log \e$. Hence, every solution which uses indexing already incurs a redundancy of at best roughly $M\log \e$ bits. However, the indexing construction is asymptotically optimal with increasing $L$ \cite{HRT17}. Note that the suboptimality of indexing for multiset codes has been shown in \cite{KT18}.

\item While the value of $L$ is moderate, e.g., in the order of a few hundreds, $M$ is significantly larger. In this work we assume that $L=o(M)$ and usually $M$ can be polynomial in $L$ or exponential in $\beta L$ for some $0<\beta<1$. In any event we require $M\leq 2^L$.

\item We present the results in this work for binary sequences, however most or all of them can be extended to the non-binary case (and, in particular, the quaternary case).

\end{enumerate}
\end{addmargin}

\end{remark}

The results about the redundancy of the proposed constructions and their lower bounds are summarized in Table \ref{tab:construction}.
\vspace{-.3cm}\begin{table}[h]
	\caption{Redundancy of Constructions and Bounds}\vspace{-1ex}
	\centering
	\begin{tabular}{l|C{4.2cm}|c}
		Error correction & Construction & approx. Bound \\
		\hline
		\multirow{3}{*}{$(s,t,L)_{\mathcal{H}}$ or}& $c_MM + (s+2t) (L-\lceil \log M \rceil)$ & \multirow{3}{*}{$(s+t)L+$}  \\ \cline{2-2}
		\multirow{3}{*}{$(s,t,L)_{\mathcal{L}}$}& $(s + 2 t)L$ & \multirow{3}{*}{$t\log M$}  \\ \cline{2-2}
		&  $M^c (s+2t) (L-\log M + \log \e) + M^{1-c} \log \big(\e M^{\frac{c}{2}} \big) - (s+2t) \log \e$ &  \\ \hline
		$(0,1,1)_\mathcal{L}$ & $\log(L+1)$ & $\log(L)-1$ \\ \hline
		$(0,1,1)_\mathcal{H}$ & $2L$ & $\log L$ \\ \hline
		$(0,M,1)_\mathcal{L}$ & $M\log(L+1)$ & $M (\log L -1) $ \\ \hline
		$(0,M,\epsilon)_\mathcal{H}$ & $M\epsilon\lceil \log L \rceil$ & $M\epsilon \log (L/\epsilon)$ \\
	\end{tabular}
	\label{tab:construction}
\end{table}\vspace{-.45cm}

\section{Code Constructions}\label{sec:const}

\subsection{An Index-Based Construction}
The following construction is based on adding an index in front of all sequences $\ve{x}_i$ and using a maximum distance separable (MDS) code over the $M$ sequences. For all $n$ and~$k$, where $k\leq n$ we denote by $\mathsf{MDS}[n,k]$ an MDS code over any field of size at least $n-1$. For all $1\leq i\leq M$ we will use $\ve{I}(i) \in \mathbb{Z}_2^{\lceil \log M \rceil}$ to denote the binary representation of the index $i$ with $\lceil \log M\rceil$ bits. 

In Construction~\ref{con:index:rs}, the sequences $\ve{x}_i = (\ve{I}(i), \ve{u}_i)$ of each codeword set are constructed by writing a binary representation of the index, $\ve{I}(i)$, of length $\lceil \log M \rceil$ in the first part of each sequence. Then, the remaining part $\ve{u}_i$ is viewed as a symbol over the extension field $\mathbb{F}_{2^{L-\lceil \log M \rceil}}$, and $(\ve{u}_1,\ldots,\ve{u}_M)$  will form a codeword in some MDS code\footnote{Note that we assume $M \leq \sqrt{2^L}$ in this section to guarantee the existence of the MDS code. However, the case $M > \sqrt{2^L}$ can always be addressed by employing non-MDS codes.}. In this construction and in the rest of the paper whenever we write the set $\mathcal{S}$ we assume it is a set of $M$ sequences denoted by $\mathcal{S} =  \{ \mathbf{x}_1, \ldots,\mathbf{x}_M \} \in \mathcal{X}_M^L$. 
The following construction is based on the findings in \cite{HRT17}, where index-based constructions are analyzed for the correction of losses only.
\begin{construction} \label{con:index:rs}
For all $M, L$, and a positive integer $\delta$, let $\mathcal{C}_1 (M,L,\delta)$ be the code defined by \vspace{-0.5ex}
	\begin{align*}
	\mathcal{C}_1 (M,L,\delta)= \{ \mathcal{S} \in \mathcal{X}_M^L : \, & \ve{x}_i = (\ve{I}(i), \ve{u}_i),  \\ 
	&(\ve{u}_1, \ldots, \ve{u}_M) \in \mathsf{MDS}[M,M-\delta] \}.
	\end{align*}
\end{construction}

\begin{lemma}
For all $M,L,\delta$, the code $\mathcal{C}_1 (M,L,\delta)$ is an  $(s,t,L)_\mathcal{H}$ error-correcting code for all $s+2t \leq \delta$.
\end{lemma}
\begin{proof}
	To begin with, we observe that if we can recover the MDS codeword $\ve{U} = (\ve{u}_1, \ve{u}_2, \dots, \ve{u}_M)$, we can also recover $\mathcal{S}$. Given $\mathcal{S}'$, we create the received word $\ve{U}'$ by declaring 
position $j$ to be an erasure if  the index $\ve{I}(j)$ does not appear or appears more than once in $\mathcal{S}'$. The remaining positions in $\ve{U}'$ are filled with the corresponding symbols $\ve{u}'_j$. We will show that the number of erasures $s'$ and the number of errors $t'$ in $\ve{U}'$ satisfy $s' + 2t' \leq \delta$. Denote by $\mathcal{U}_I, \mathcal{L}_I,\mathcal{F}_I$ the sets of indices (first $\lceil \log M \rceil$ bits) of sequences in $\mathcal{S}$ corresponding to $\mathcal{U}$, $\mathcal{L}$, or $\mathcal{F}$, respectively. Further, $\mathcal{F}_I'$ is the set of indices of received sequences in $\mathcal{S}'$ that are the erroneous outcomes of the sequences in $\mathcal{F}$. First, we have $s' \leq s+t-t' + |\mathcal{F}_I' \cap \mathcal{U}_I|$ where $|\mathcal{F}_I' \cap \mathcal{U}_I|$ accounts for the situation when an erroneous sequence has the same index as an error-free one. Secondly, the number of errors satisfies $t' \leq |\mathcal{F}_I' \cap (\mathcal{F}_I \cup \mathcal{L}_I)|$. Hence, $s'+2t' \leq s+t+t' + |\mathcal{F}_I' \cap \mathcal{U}_I| \leq s+2t \leq \delta$.
\end{proof}
Similarly we obtain the error-correction capability of Construction \ref{con:index:rs} with respect to insertion and deletion errors.
\begin{lemma}
	For all $M,L,\delta$, the code $\mathcal{C}_1(M,L,\delta)$ is an $(s,t,L)_\mathcal{L}$ error-correcting code for all $s+2t \leq \delta$.
\end{lemma}
Note that for the practically important case of losses and combinations of substitution and deletion errors, $\mathcal{C}_1(M,L,\delta)$ is error-correcting, if $s+2t_\mathcal{H} + t_\mathcal{D}\leq \delta$, where $t_\mathcal{H}$ is the number of sequences suffering from substitution errors only and $t_\mathcal{D}$ is the number of sequences with deletion errors. The same also holds for combinations of substitution and insertion errors. For all $M,L,\delta$, the redundancy of the code $\mathcal{C}_1 (M,L,\delta)$ is
$$ r(\mathcal{C}_1 (M,L,\delta)) = c_MM + \delta (L-\lceil \log M \rceil).$$

\subsection{A Construction Based On Constant Weight Codes}

Imposing an ordering (e.g., lexicographic) onto the sequences in $\mathbb{Z}_2^L$, every data set $S \in \mathcal{X}_M^L$ can be represented by a binary vector $\mathbf{v}(\mathcal{S})$ of length $2^L$, where each non-zero entry in $\mathbf{v}(\mathcal{S})$ indicates that a specific sequence is contained in the set $\mathcal{S}$. The possible data sets can therefore be represented\footnote{This representation has been used as a proof technique in \cite{HRT17}.} by the set of constant-weight binary vectors of length $2^L$,
\[ \mathcal{V}_M^L = \{ \mathbf{v} \in \{0,1\}^{2^L} : \mathrm{wt}(\mathbf{v}) = M \}, \]
where $\mathrm{wt}(\mathbf{v})$ denotes the \textit{Hamming weight} of $\mathbf{v}$, i.e., the number of non-zero entries inside the vector $\mathbf{v}$. Using this representation, a sequence loss corresponds to an asymmetric $1 \rightarrow 0$ error inside $\ve{v}(\mathcal{S})$. Errors inside a sequence are either single errors in the Johnson graph, see e.g. \cite{BSSS90}, or single asymmetric $1 \rightarrow 0$ errors, if the erroneous sequence coincides with an already present sequence in $\mathcal{S}$. This suggests the following construction.
\begin{construction} \label{con:constant:weight}
For all $M,L$ and positive integers $s,t$, let $\mathcal{C}_M^L(s,t)\subseteq \mathcal{V}_M^L$ be an $M$-constant-weight code of length $2^L$, which corrects $s$ asymmetric $1 \rightarrow 0$ errors and $t$ errors in the Johnson graph. We then define the following code
	$$ \mathcal{C}_2(M,L,s,t) = \{ \mathcal{S} \in \mathcal{X}_M^L : \ve{v}(\mathcal{S}) \in \mathcal{C}_M^L(s, t) \}. $$
\end{construction}
\begin{lemma}
For all $M,L$ and positive integers $s,t$, the code $\mathcal{C}_2(M,L,s,t) $ is an $(s, t, L)_\mathcal{H}$ error-correcting code.
\end{lemma}
\begin{proof}
	Denote by $\mathcal{S}'$ the received set after at most $s$ losses of sequences and errors in at most $t$ sequences. Let $s'$ be the number of asymmetric errors and $t'$ be the number of errors in $\ve{v}(\mathcal{S})$ with $s'+t' \leq s+t$ and $t'\leq t$. Note that $s' = M - \mathrm{wt}(\ve{v}(\mathcal{S}'))$ is detectable by the decoder. If $s' \leq s$, then the decoder can directly decode $s'\leq s$ losses and $t'\leq t$ errors in the Johnson graph. If $s' > s$, the decoder adds $s'-s$ (arbitrarily placed) ones to $\ve{v}(\mathcal{S}')$, resulting in exactly $s$ losses and at most $t'+s'-s \leq t$ errors in the Johnson graph.
\end{proof}
Since asymmetric and Johnson graph errors can be represented by one, respectively two substitutions, we can use an $M$-constant-weight subset of any standard error correcting code, which corrects $\tau=s+2t$ errors for $\mathcal{C}_M^L(s,t)$. In \cite[ch.~5.5]{Rot06} it is shown that a $\tau$-error-correcting binary alternant code code of length $2^L$ has dimension at least $2^L-\tau L$. Due to the pigeonhole principle, there is one coset of the alternant code that contains at least $\binom{2^L}{M}\big/2^{\tau L}$ words with constant weight $M$. Hence, there exists a code $\mathcal{C}_M^L(s,t)$, such that
$$ r(\mathcal{C}_2(M,L,s,t)) \leq (s + 2 t)L. $$
This redundancy is smaller than the redundancy of Construction \ref{con:index:rs}, especially for the case $L = o(M)$. 

\subsection{An Improved Indexed-Based Construction}
Construction~\ref{con:index:rs}, which uses indexing, is beneficial for its simplicity in the encoding and decoding procedure, however its redundancy is larger than that of Construction~\ref{con:constant:weight}. On the other hand, Construction~\ref{con:constant:weight} does not provide an efficient encoder and decoder. In this section, we present a construction which introduces ideas from both of these methods. 

The main idea of this construction is to reduce the number of bits allocated for indexing each sequence. This allows a trade-off in redundancy with respect to $L$ and $M$. To simplify notation, we assume here that $M = 2^z$ for some $z \in \mathbb{N}$.
\begin{construction} \label{con:reduce:index}
Denote by $\ve{I}_c(i) \in \mathbb{Z}_2^{(1-c) \log M}$ the $(1-c) \log M$ most significant bits of the binary representation $\ve{I}(i)$ of $i$, where $0\leq c < 1$ and $c \log M \in \mathbb{N}_0$. Further, for $0\leq i\leq M^{1-c}-1$, let $\ve{U}_i = \{ \ve{u}_{iM^c+1}, \dots, \ve{u}_{(i+1)M^c} \}$ denote a set of distinct sequences with the same index $\ve{I}_c(i)$, which are ordered lexicographically and form a symbol over a field. For $\delta \geq 0$, let $\mathcal{C}_3 (M,L,c,\delta) $ be the code defined by 
	\begin{align*}
	\mathcal{C}_3 (M,L,c,\delta) &= \{  \mathcal{S}  \in \mathcal{X}_M^L : \, \ve{x}_i = (\ve{I}_c(i), \ve{u}_i) , \\ &(\ve{U}_1, \dots, \ve{U}_{M^{1-c}}) \in \mathsf{MDS}[M^{1-c},M^{1-c}-\delta] \}. 
	\end{align*}
\end{construction}

Note that there are $M^{1-c}$ groups of sequences which use the same index. \footnote{The symbols of the MDS code are symbols of a finite field with size $\binom{2^LM^{c-1}}{M^c}$ and we therefore require $M^{1-c} \leq \binom{2^LM^{c-1}}{M^c}$.}
\begin{lemma} \label{lemma:reduce:index}
For all $M,L,c,\delta$, the code $\mathcal{C}_3(M,L,c,\delta)$ is an $(s,t,L)$ error-correcting code for all $s+2t \leq \delta$.
\end{lemma}
Lemma \ref{con:index:rs} is proven similiar to Lemma \ref{lemma:reduce:index}. The redundancy of $\mathcal{C}_3$ can be shown to be approximately
$$ r(\mathcal{C}_3) \approx M^c \delta (L-\log M + \log \e) + M^{1-c} \log \left(\e M^{\frac{c}{2}} \right) - \delta \log \e. $$
\subsection{Special Constructions}
We begin with an observation about the equivalence of $(0,t,\epsilon)_\mathcal{L}$ codes for the case where there are either only insertion or only deletion errors inside the sequences.
\begin{lemma}
	A code $\mathcal{C} \subseteq \mathcal{X}_M^L $ is $(0,t,\epsilon)$ insertion-only correcting if and only if it is $(0,t,\epsilon)$ deletion-only correcting.
\end{lemma}
Note that a $(0,t,\epsilon)_\mathcal{L}$ deletion-only (or insertion-only) code, with $\epsilon \geq  2$,  is in general not insertion \textit{and} deletion correcting. A counterexample is the code $\mathcal{C} = \{\{0000,1111,1000\} , \{0000,1111,0111\}\}$, which is both $(0,1,2)$ insertion-only and deletion-only correcting, but not $(0,1,2)_\mathcal{L}$ insertion and deletion correcting. 

The following construction is based on Varshamov-Tenengolts (VT) codes \cite{VT65,Lev66} that correct a single insertion/deletion in one of the $M$ sequences. This code can be extended to an arbitrary alphabet size $q$ by applying non-binary VT codes \cite{Ten84}. The construction employs the idea of using single-erasure-correcting code over the checksums. The insertion/deletion can then be corrected using the corresponding checksum. Note that this idea is similar to the concept of tensor product codes \cite{Wol06}. 
\begin{defn}
	The checksum $s_L(\ve{x})$ of $\ve{x} \in \mathbb{Z}_2^L$ is defined by\vspace{-1ex}$$ s_L(\ve{x}) = \sum_{i=1}^{L} ix_i \bmod (L+1). $$
\end{defn}
\begin{construction} \label{con:single:insertion}
	For an integer $a$, with $0\leq a \leq L$, the code construction $\mathcal{C}_4(M,L,a) $ is given by
	\begin{align*}
	\mathcal{C}_4(M,L,a) = \bigg\{ \mathcal{S} \in \mathcal{X}_M^L :\sum_{i=1}^{M} s_L(\ve{x}_i) \equiv a \bmod (L+1) \bigg\}.
	\end{align*}
\end{construction}
\begin{lemma}
For all $M,L,a$, the code $\mathcal{C}_4(M,L,a)$ is a $(0,1,1)_\mathcal{L}$ error-correcting code.
\end{lemma}
\begin{proof}
	Assume there has been a single insertion or deletion in the $k$-th sequence. After the reading process, the $M-1$ error-free sequences can be identified as they have length exactly $L$. The checksum deficiency is given by
	\[ a - \sum_{i \in \mathcal{U}} s_L(\ve{x}_i) \bmod (L+1) = s_L(\ve{x}_k). \]
	The error in $\ve{x}_k$ is corrected by decoding into the VT code with checksum $s_L(\ve{x}_k)$.
\end{proof}
Based on the pigeonhole principle there exists $0\leq a \leq L$ such that the redundancy of the code $\mathcal{C}_4(M,L,a)$ satisfies
$$ r(\mathcal{C}_4(M,L,a)) \leq \log (L+1). $$
As we will show in Theorem \ref{thm:bound:levenshtein}, the redundancy of any $(0,1,1)_\mathcal{L}$ error-correcting code is at least $\log(L+2)-1$, and thus Construction \ref{con:single:insertion} is close to optimal. 

Using VT codes, we propose another construction of $(0,M,1)_\mathcal{L}$ error-correcting codes. That is, the code can correct a single deletion or insertion in every sequence.
\begin{construction} \label{con:multiple:insertion}
	Let $a \in \mathbb{N}_0$, with $0\leq a \leq L$. Then,
	$$\mathcal{C}_5(M,L,a) \hspace{-0.25ex}=  \hspace{-0.25ex}\{ \mathcal{S} \hspace{-0.25ex} \in \hspace{-0.25ex} \mathcal{X}_M^L  \hspace{-0.25ex}: \hspace{-0.25ex} s_L(\ve{x}_i)  \hspace{-0.25ex}\equiv \hspace{-0.25ex} a \bmod (L \hspace{-0.15ex} + \hspace{-0.15ex}1),\forall \, 1 \hspace{-0.25ex}\leq  \hspace{-0.25ex}i  \hspace{-0.25ex}\leq  \hspace{-0.25ex}M \}.$$
\end{construction}
\begin{lemma}
	The code $\mathcal{C}_5(M,L,a)$ is a $(0,M,1)_\mathcal{L}$ error-correcting code.
\end{lemma}
By Construction \ref{con:multiple:insertion}, all sequences $\ve{x}_i$ have the same checksum $a$, which allows to correct single insertions or deletions in each sequence. It is known \cite{Lev66} that the number of words satisfying $s_L(\ve{x}) = 0 \bmod (L+1)$ is at least $2^L/(L+1)$. Therefore the redundancy of Construction \ref{con:multiple:insertion} is at most
\begin{align*}
	r(\mathcal{C}_5(M,L,0))  \leq   M \left(\log(L+1) +\hspace{-0.25ex} \frac{M\log \e}{2^L/(L+1)-M}\right).
\end{align*}
With our assumption $M=2^{\beta L}$, we obtain a redundancy of $r(\mathcal{C}_5(M,L,0)) \approx M \log(L+1)$. The next construction can be used to correct $\epsilon$ substitution errors in each sequence. Let $\mathcal{C}[L,\epsilon]$ a binary $\epsilon$-error-correcting code of length $L$.
\begin{construction} \label{con:hamming} For all $M,L$, and $\epsilon$ we define the code
	$$\mathcal{C}_6(M,L,\epsilon) = \{ \mathcal{S} \in \mathcal{X}_M^L : \mathcal{S} \subseteq \mathcal{C}[L,\epsilon]\},$$
\end{construction}
\begin{lemma}
	The code $\mathcal{C}_6(M,L,\epsilon)$ is a $(0,M,\epsilon)_\mathcal{H}$ error-correcting code.
\end{lemma}
The proof is immediate, since every sequence is a codeword of a code that can correct $\epsilon$ errors. For $\mathcal{C}[L,\epsilon]$ we use a binary $\epsilon$-error correcting alternant code of length $L$, which has redundancy at most $\epsilon \lceil \log L \rceil$ \cite[ch. 5.5]{Rot06} and thus obtain a code $\mathcal{C}_6(M,L,\epsilon)$ with redundancy at most
	$$ r(\mathcal{C}_6(M,L,\epsilon))  \leq   M\left(\epsilon \lceil\log L\rceil +  \frac{M \log \e}{2^{L-\epsilon \lceil \log L \rceil} -M} \right),$$
if $M \leq 2^{L-\epsilon \lceil \log L \rceil}$. With our assumption $M = 2^{\beta L}$, the redundancy is roughly $r(\mathcal{C}_6(M,L,\epsilon)) \approx M\epsilon \lceil\log L\rceil$.

\section{Upper Bounds}
In this section we derive non-asymptotic sphere packing upper bounds on codes within the presented storage model.
\begin{defn}
	The \textit{error ball} $B^\mathcal{H}_{t,\epsilon}(\mathcal{S})$ $[B^\mathcal{L}_{t,\epsilon}(\mathcal{S})]$ is defined to be the set of all possible received sets $\mathcal{S}' = \mathcal{U} \cup \mathcal{F}'$ after $t$ (or fewer) sequences of $\mathcal{S} \in \mathcal{X}_M^L$ have been distorted by $\epsilon$ (or fewer) substitution [insertion/deletion] errors each.
\end{defn}
\begin{defn}
	The \textit{error ball} $B^\mathcal{H}_{\epsilon}(\ve{x})$ $[B^\mathcal{L}_{\epsilon}(\ve{x})]$ around $\ve{x} \in \mathbb{Z}_2^L$ is defined to be the set of all possible received vectors $\ve{x}'\neq \ve{x}$, after $\epsilon$ (or fewer) substitutions [insertions/deletions].
\end{defn}
\begin{thm} \label{thm:bound:hamming} The cardinality of any $(0,t,\epsilon)_\mathcal{H}$ error-correcting code $\mathcal{C} \subseteq \mathcal{X}_M^L$ satisfies\vspace{-1ex}
	$$|\mathcal{C}| \leq \frac{\sum_{i=M-t}^{M}\binom{2^L}{i}}{(B^\mathcal{H}_\epsilon-(t-1)N^\mathcal{H}_\epsilon)^t},$$
	where $B^\mathcal{H}_\epsilon = \sum_{i=1}^{\epsilon} \binom{L}{i}$ is the size of the $\epsilon$-error ball and $N_\epsilon^\mathcal{H} = \sum_{i=0}^{\epsilon-1}\binom{L-1}{i}$ is the maximum intersection size of two $\epsilon$-error balls around two distinct words.
\end{thm}
\begin{proof}
	We derive a lower bound on $|B^\mathcal{H}_{t,\epsilon}(\mathcal{S})|$. To each of the $t$ erroneous sequences we can associate a unique set of at least $B^\mathcal{H}_\epsilon-(t-1)N^\mathcal{H}_\epsilon$ distinct words in the substitution error ball. This is because there are $B^\mathcal{H}_\epsilon$ elements in the substitution ball and there are at most $N^\mathcal{H}_\epsilon$ elements in common with each of the $t-1$ other erroneous sequences. Therefore, we get $(B^\mathcal{H}_\epsilon-(t-1)N^\mathcal{H}_\epsilon)^t$ possible unique received sets. The nominator counts all possible received sets of size $M-t$ to $M$, which yields the bound by a sphere packing argument. The value for $N_\epsilon^\mathcal{H}$ is known from \cite{Lev01}.
\end{proof}
Using this bound yields for small $t$ and $\epsilon=1$ a minimum redundancy of approximately $t \log(L)$.
\begin{thm} \label{thm:bound:levenshtein} The cardinality of any $(0,t,\epsilon)_\mathcal{L}$ error-correcting code $\mathcal{C} \subseteq \mathcal{X}_M^L$ satisfies
	$$|\mathcal{C}| \leq \frac{\binom{2^L}{M-t}\binom{2^{L+\epsilon}}{t}}{\binom{M}{t}(S_\epsilon^\mathcal{I} - (t-1) N_\epsilon^\mathcal{I} )^t},$$
	where $S^\mathcal{I}_\epsilon = \sum_{i=0}^{\epsilon} \binom{L+\epsilon}{i}$ is the size of the $\epsilon$-insertion sphere and $N_\epsilon^\mathcal{I} = \sum_{i=0}^{\epsilon-1}\binom{L+\epsilon}{i}(1-(-1)^{\epsilon-i})$ is the maximum intersection of two $\epsilon$-insertion spheres of two distinct words.
\end{thm}
The proof of Theorem \ref{thm:bound:levenshtein} follows the same idea as the proof of Theorem \ref{thm:bound:hamming}. For small $t$ and $\epsilon=1$, this bound implies a minimum redundancy of approximately $t (\log (L+2)-1)$.
\subsection{Asymptotic bounds}
We now derive asymptotic bounds for large $L$ on the redundancy for $(0,M,\epsilon)_\mathcal{H}$ and $(0,M,\epsilon)_\mathcal{L}$ error-correcting codes.
\begin{lemma} \label{lemma:substitutions:subset}
	Denote by $\mathcal{Y} \subseteq \mathcal{S} \in \mathcal{X}_M^L$ the largest set such that $\mathcal{Y}$ is an $\epsilon$-substitution correcting code. Then,
	$$|B_{t,\epsilon}^\mathcal{H}(\mathcal{S})| \geq \left\{ \begin{array}{ll}
	\left(B_{\epsilon}^\mathcal{H}\right)^{|\mathcal{Y}|}, & \text{if}\, |\mathcal{Y}| \leq t \\
	\binom{|\mathcal{Y}|}{t} \left(B_{\epsilon}^\mathcal{H}\right)^{t}, & \text{else} 
	\end{array} \right.$$
	where $B_{\epsilon}^\mathcal{H} = \sum_{i=1}^{\epsilon} \binom{L}{i}$.
\end{lemma}
\begin{proof}
	In each of the distinct error balls $B_\epsilon^\mathcal{H}(\ve{x})$, $\ve{x} \in \mathcal{Y}$ we have at least $B_{\epsilon}^\mathcal{H} = |B_{\epsilon}^\mathcal{H}(\ve{x})|$ possible patterns of unique outcomes for $B_{t,\epsilon}^\mathcal{H}(\mathcal{S})$ by either adding an error to $\ve{x}$ such that a sequence in $B_\epsilon^\mathcal{H}(\ve{x}) \setminus \mathcal{S}$ is obtained or by adding an error to a sequence in $B_\epsilon^\mathcal{H}(\ve{x}) \cap \mathcal{S}$ such that $\ve{x}$ is obtained.
\end{proof}
\begin{thm} \label{thm:substitutions:asymptotic} The redundancy of any $(0,M,\epsilon)_\mathcal{H}$ error-correcting code $\mathcal{C} \subseteq \mathcal{X}_M^L$ satisfies asymptotically
	$$ r(\mathcal{C}) \gtrsim c M \log (B_\epsilon^\mathcal{H}), $$
	for any $0\leq c<1$, when $L \rightarrow \infty$ and $M = 2^{\beta L}$, $0<\beta <1$.
\end{thm}

\begin{proof}
	Denote by $D(c)$ the number of data words $\mathcal{S} \in \mathcal{X}_M^L$ which have a ball size $|B_{M,\epsilon}^\mathcal{H}(\mathcal{S})| < (B_\epsilon^\mathcal{H})^{cM}$, where $0\leq c < 1$. By Lemma \ref{lemma:substitutions:subset}, $D(c)$ is at most the number of data sets, which do not contain an $\epsilon$-error correcting code $\mathcal{Y} \subseteq \mathcal{S}$ of size at least $ct$. By a sphere packing argument, it follows that any $(0,M,\epsilon)_\mathcal{H}$ correcting code $\mathcal{C} \subseteq \mathcal{X}_M^L$ satisfies
	$$ |\mathcal{C}| \leq \frac{\sum_{i=cM}^{M}\binom{2^L}{i}}{(B_\epsilon^\mathcal{H})^{cM}} + D(c). $$
	It can be shown that the first term in this sum dominates the bound for all $0\leq c < 1$, when $M = 2^{\beta L}$, with $0<\beta < 1$.
\end{proof}
\begin{thm} \label{thm:insertions:asymptotic} The redundancy of any $(0,M,\epsilon)_\mathcal{L}$ error-correcting code $\mathcal{C} \subseteq \mathcal{X}_M^L$ satisfies asymptotically
	$$ r(\mathcal{C}) \gtrsim c M (\log (S_\epsilon^\mathcal{I})-\epsilon), $$
	for any $0\leq c<1$, when $L \rightarrow \infty$ and $M = 2^{\beta L}$, $0<\beta <1$.
\end{thm}
Theorem \ref{thm:insertions:asymptotic} can be shown by noting that most balls $B_{M,\epsilon}^\mathcal{L}(\mathcal{S})$ have size at least $(S_\epsilon^\mathcal{I})^{cM}$, similar to the proof of Theorem \ref{thm:substitutions:asymptotic}.
\subsection{Bound for losses and errors}
\begin{thm}
	The redundancy of any $(s,t,L)_\mathcal{H}$ or $(s,t,L)_\mathcal{L}$ correcting code $\mathcal{C} \subseteq \mathcal{X}_M^L$ satisfies
	$$ r(\mathcal{C}) \geq (s+t) \log (2^L-M-t) + t\log (M-s-t) - \log (t!(s+t)!).$$
\end{thm}
\begin{proof}
	Choosing $s+t$ sequences to be erroneous and letting each of the $t$ erroneous ones be one of the $2^L-M$ sequences in $\mathcal{X}_M^L \setminus \mathcal{S}$, we can use a sphere-packing argument to show that any $(s,t,L)_\mathcal{H}$ or $(s,t,L)_\mathcal{L}$ correcting code $\mathcal{C} \subseteq \mathcal{X}_M^L$ satisfies
	$|\mathcal{C}| \leq {\binom{2^L}{M-s}}\big/{\big(\binom{M}{t+s} \binom{2^L-M}{t}\big)}.$
\end{proof}

\bibliography{IEEEabrv,ref.bib}
\bibliographystyle{IEEEtranS}

\end{document}